\documentclass[11pt,journal,onecolumn,draftclsnofoot]{IEEEtran}

\usepackage[tbtags]{amsmath}
\usepackage{amssymb,amsthm}
\usepackage{verbatim}
\usepackage{amsxtra}
\usepackage{enumerate}
\usepackage{amsfonts}
\usepackage{color}
\usepackage{cite}
\usepackage{float}
\usepackage{epsfig, graphicx, psfrag, subfig}
\usepackage{units}
\usepackage[table]{xcolor}
\usepackage{calc}
\usepackage{amsfonts,amssymb}
\usepackage{epstopdf}
\usepackage[mathcal]{euscript}
\usepackage{multirow}
\usepackage{arydshln}
\usepackage{url}
\usepackage{caption}
\usepackage{cite}
\usepackage{mathtools}
\mathtoolsset{showonlyrefs=true}
\usepackage{pstool}

\usepackage{tikz}
\usetikzlibrary{arrows}
\usepackage{verbatim}
\usepackage{relsize}
\newtheorem{corollary*}{Corollary}
\newtheorem{remark}{Remark}
\newtheorem{example}{Example}
\newtheorem{lemma}{Lemma}
\newtheorem*{lemma*}{Lemma}

\begin{document}

\title{An Inequality for the Correlation of Two Functions Operating on Symmetric Bivariate Normal Variables}

\author{Ran~Hadad, Uri~Erez, and Yaming~Yu
\thanks{R.~Hadad and U.~Erez with the Department of Electrical Engineering-Systems, Tel Aviv University, Ramat Aviv 69978, Israel;
   (e-mails: ranhadad@post.tau.ac.il, uri@eng.tau.ac.il).}
\thanks{Y.~Yu is with the Department of Statistics, University of California, Irvine, CA 92697, USA;
(e-mail: yamingy@uci.edu).}
}
\maketitle

\begin{abstract}
An inequality is derived for the correlation of two univariate functions operating on symmetric bivariate normal random variables. The inequality is a simple consequence of the Cauchy-Schwarz inequality.
\end{abstract}
\section{Introduction}
Statistical characterization of the output of non-linear systems operating on stochastic processes is in general difficult.
Nonetheless, when the input process is Gaussian and the system is a memoryless non-linearity, several particularly simple and useful properties are known. Among these are Bussgang's theorem ~\cite{Bussgang} and its generalizations (e.g., \cite{price1958useful,mcgee1969circularly}), results concerning the maximal correlation coefficient  \cite{lancaster1957some}, as well as results on the output distortion-to-signal power ratio \cite{rowe1982memoryless}.
In the present note, we describe another simple result as described in the next section.
A generalization to a more general class of random variables is described in Section~\ref{sec:generalization}.
An application is presented in Section~\ref{sec:app}.

\section{Statement of Result and Proof}
\begin{lemma}
\label{lemma1}
Let $Z_1$ and $Z_2$ be zero-mean bivariate normal random variables with variance $\sigma^2$ and correlation coefficient $\rho > 0$.
Then,
\begin{align*}
\mathbb{E}^2  \left[ g_1(Z_1) g_2(Z_2) \right] \leq \mathbb{E}  \left[ g_1(Z_1) g_1(Z_2) \right] \mathbb{E}\left[ g_2(Z_1) g_2(Z_2) \right],
\end{align*}
 for any $g_1$ and $g_2$ for which the expectations exist, with equality if and only if $g_1$ and $g_2$ are equal up to a multiplicative constant.
\end{lemma}

Note that for $\rho=1$ the lemma reduces to the the standard probabilistic Cauchy-Schwarz inequality.
\begin{proof}
Consider the inner product of real functions defined by
\begin{align*}
\langle f,g\rangle \triangleq \mathbb{E}\left[f(Z)g(Z)\right]
\end{align*}
where $Z\sim \mathcal{N}(0,\sigma^2)$.

According to Mehler's formula~\cite{mehler1866ueber} (see also \cite{kibble1945extension}) the joint density function of $(Z_1,Z_2)$ may be written as,
\begin{align}\label{jointexpansion}
f_{Z_1,Z_2}(z_1,z_2) &= \frac{1}{2\pi\sigma^2\sqrt{1-\rho^2}}  e^{-\frac{1}{2}\frac{z_1^2-2 \rho z_1 z_2 + z_2^2}{\sigma^2(1-\rho^2)}} \nonumber\\
&= \frac{1}{2\pi\sigma^2}  e^{-\frac{1}{2}\frac{z_1^2+z_2^2}{\sigma^2}} \sum_{n=0}^{\infty}\frac{1}{n!}H\!e_n\big(\tfrac{z_1}{\sigma}\big)H\!e_n\big(\tfrac{z_2}{\sigma}\big)\rho^n,
\end{align}
where $H\!e_n(x)$ are the probabilists' Hermite polynomials defined as,
\begin{align*}
H\!e_n(x) = (-1)^ne^{\frac{x^2}{2}}\frac{d^n}{dx^n}e^{-\frac{x^2}{2}}, \quad n\geq0.
\end{align*}

The Hermite polynomials constitute a complete orthogonal basis of polynomials with respect to the standard normal probability density function
~\cite{szeg1939orthogonal}, so that
\begin{align*}
\langle \widetilde{H}\!e_{n},\widetilde{H}\!e_{m} \rangle = \mathbb{E}\big[\widetilde{H}\!e_{n}(Z)\widetilde{H}\!e_{m}(Z)\big] = n!\,\delta_{n,m}
\end{align*}
where $\widetilde{H}\!e_{n}(x) = H\!e_{n}(\tfrac{x}{\sigma})$, and $\delta_{n,m}$ is the Kronecker delta function.

Let $a_{g_i,n} = \langle g_i,\widetilde{H}\!e_n\rangle$. The function $g_i(x)$ may be represented by the series
\begin{align}\label{gexpansion}
  g_i(x) = \sum\limits_{n=0}^{\infty}\frac{1}{n!}a_{g_i,n}H\!e_n(\tfrac{x}{\sigma})\;,\;i = 1,2.
\end{align}
The following expectations are obtained by applying~\eqref{jointexpansion}:
\begin{align*}
&\mathbb{E}\left[g_1(Z_1) g_2(Z_2)\right]=\sum\limits_{n=0}^{\infty}\frac{1}{n!}a_{g_1,n}a_{g_2,n}\rho^n, \\
&\mathbb{E}\left[g_i(Z_1) g_i(Z_2)\right]=\sum\limits_{n=0}^{\infty}\frac{1}{n!}a_{g_i,n}^2\rho^n\;,\;i = 1,2.
\end{align*}
Using these identities and assuming $\rho>0$ we have,
\begin{align}
\mathbb{E}^2\left[g_1(Z_1) g_2(Z_2)\right] &= \left(\sum\limits_{n=0}^{\infty}\frac{a_{g_1,n}a_{g_2,n}}{n!}\rho^n\right)^2 \nonumber \\
&=\left(\sum\limits_{n=0}^{\infty}\frac{a_{g_1,n}\,\rho^{\tfrac{n}{2}}}{\sqrt{n!}}\frac{a_{g_2,n}\,\rho^{\tfrac{n}{2}}}{\sqrt{n!}}\right)^2\nonumber \\
&\leq\left(\sum\limits_{n=0}^{\infty}\frac{a_{g_1,n}^2\,\rho^n}{n!}\right)\left(\sum\limits_{n=0}^{\infty}\frac{a_{g_2,n}^2\,\rho^{n}}{n!}\right)\label{sqaured}\\
&=\mathbb{E}\left[g_1(Z_1) g_1(Z_2)\right]\mathbb{E}\left[g_2(Z_1) g_2(Z_2)\right]\nonumber,
\end{align}
where the inequality follows by the Cauchy–Schwarz inequality for sequences, which holds with equality only when 
$a_{g_1,n} = c\cdot a_{g_2,n}$ for some constant $c$, $\rho >0$, and for all $n\geq0$. Since the case of equality holds only when the coefficients in the series~\eqref{gexpansion} are equal up to a multiplicative constant, it follows that equality holds only when $g_1$ and $g_2$ are equal up to a multiplicative constant.
\end{proof}


When both functions are even or odd, we may apply the lemma to $Z_1$ and $-Z_2$ to obtain,
\begin{corollary*}
Let $Z_1$ and $Z_2$ be zero-mean bivariate normal random variables with variance $\sigma^2$ and correlation coefficient $\rho\neq0$.
Then, for $g_1$ and $g_2$ that are both even or odd functions,
\begin{align*}
\mathbb{E}^2\left[g_1(Z_1) g_2(Z_2)\right] \leq \mathbb{E}\left[g_1(Z_1) g_1(Z_2)\right] \mathbb{E}\left[g_2(Z_1) g_2(Z_2)\right],
 \end{align*}
for any such $g_1$ and $g_2$ for which the expectations exist, with equality if and only if $g_1$ and $g_2$ are equal up to a multiplicative constant.
\end{corollary*}

\begin{remark}
It is interesting to contrast the lemma with the maximal correlation property of normal vectors.
Specifically, consider the case where $Z_1$ and $Z_2$ are zero mean and both $g_1$ and $g_2$ are odd, so that $\mathbb{E}\left[g_i(Z_j)\right]=0$, for all $i=1,2$ and $j=1,2$.
Then the maximal correlation property \cite{lancaster1957some} yields the inequality
\begin{align*}
\mathbb{E}^2  \left[ g_1(Z_1) g_2(Z_2) \right] \leq \rho^2 \mathbb{E}  \left[ g_1^2(Z_i)  \right] \mathbb{E}\left[ g_2^2(Z_i) \right].
\end{align*}
\end{remark}

\section{Generalization}
\label{sec:generalization}
The lemma may be generalized to a broader class of random variables $(Z_1,Z_2)$, where $Z_1$ and $Z_2$ result from passing some random variable $Z$ through two independent realizations of the same ``channel" (in information-theoretic terms). This generalization is stated in the next lemma. We note however that unlike Lemma~\ref{lemma1}, the ``only if'' condition for equality does not necessarily apply.
\begin{lemma}
\label{lemma_generalization}
Let $Z_1$ and $Z_2$ be random variables such that for some random variable $Z$, 
$Z_1$ and $Z_2$ are independent and identically distributed given $Z$.
Then,
\begin{align*}
\mathbb{E}^2  \left[ g_1(Z_1) g_2(Z_2) \right] \leq \mathbb{E}  \left[ g_1(Z_1) g_1(Z_2) \right] \mathbb{E}\left[ g_2(Z_1) g_2(Z_2) \right],
\end{align*}
 for any $g_1$ and $g_2$ for which the expectations exist, with equality if (but not necessarily only if) $g_1$ and $g_2$ are equal up to a multiplicative constant.
\end{lemma}

\begin{proof}
Denote for $i=1,2$,
\begin{align*}
h_i(z) &=\mathbb{E}\left[g_i(Z_1)|Z=z\right] \\
       &=\mathbb{E}\left[g_i(Z_2)|Z=z\right]
\end{align*}
since $Z_1$ and $Z_2$ are i.i.d. given $Z$. Then,
\begin{align*}
\mathbb{E}\left[g_1(Z_1)g_2(Z_2)\right]&= \mathbb{E}\left[  \mathbb{E}\left[g_1(Z_1)g_2(Z_2)|Z \right]\right]   \\
       &=\mathbb{E}\left[h_1(Z) h_2(Z) \right].
\end{align*}
Following the same steps, we also have,
\begin{align*}
\mathbb{E}\left[g_i(Z_1)g_i(Z_2)\right]&= \mathbb{E}\left[h_i^2(Z) \right]\;,\;i = 1,2.
\end{align*}
The claim now follows by applying the Cauchy-Schwarz inequality to obtain
\begin{align*}
\mathbb{E}^2\left[h_1(Z)h_2(Z)\right] \leq \mathbb{E}\left[h_1^2(Z)\right] \mathbb{E}\left[h_2^2(Z)\right].
\end{align*}
\end{proof}

\begin{example}
Let $Z$ be a vector of length $N$ whose entries are i.i.d. Bernoulli$(p)$ and similarly, let $W_1$ and $W_2$ be two independent random vectors (of length $N$) whose entries are i.i.d. Bernoulli$(q)$. Finally, let $Z_i=Z \oplus W_i$ for $i=1,2$, where $\oplus$ denotes the binary exclusive or operation.
Then the lemma holds for any two functions $g_i:\{0,1\}^N \rightarrow \mathbb{R}$.
\end{example}

\begin{remark}
We note the method that is used in \cite{lancaster1957some} to prove the maximal correlation property of bivariate normal random variables utilizes series expansions involving (probabilists') Hermite polynomials via Mehler's formula, similar to the approach taken in Lemma~\ref{lemma1}. In contrast, the proof in Lemma~\ref{lemma_generalization} follows the approach taken in \cite{yu2008maximal}, where an alternative proof to the maximal correlation property is derived.
\end{remark}


\section{Application: A Criterion for Identification of A Memoryless Non-Linearity}
\label{sec:app}
We now consider an application of Lemma~\ref{lemma1}.
Consider a memoryless non-linearity $f$ operating on a discrete-time signal corrupted by additive white Gaussian noise (AWGN), as depicted in Figure~\ref{fig:Fig1}.
\begin{figure}[h]
\begin{center}
\tikzstyle{int}=[draw, minimum size=2.2em]
\tikzstyle{sum} = [draw, circle, inner sep=0.04cm]
\begin{tikzpicture}[scale=1.05, transform shape, thick, node distance=2cm,auto,>=latex']
    \node [sum] at (1,0) (s) {+};
    \node (s1) at (s) {\scalebox{2}{+}};
    \node (x) [left of=s,node distance=1.5cm] {$x_n$};
    \path[->] (x) edge (s);
    \node (w) [above of=s,node distance=1cm] {$w_n$};
    \path[->] (w) edge (s);
    \node [int] (f) [right of=s,node distance=1.5cm] {$f(\cdot)$};
    \path[->] (s) edge node {$z_n$} (f);
    \node [int] (g) [right of=f,node distance=1.8cm] {$g(\cdot)$};
    \path[->] (f) edge node {$y_n$} (g);
    \node (v) [right of=g,node distance=1.5cm] {$\hat{z}_n$};
    \path[->] (g) edge (v);
\end{tikzpicture}
\end{center}
\caption{Memoryless non-linearity operating on a signal corrupted by AWGN.}
\label{fig:Fig1}
\end{figure}
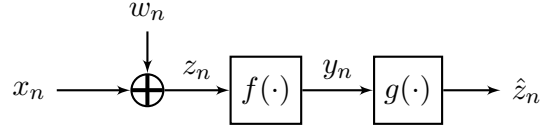 
Thus, the input $z_n$ consists of the sum of the signal $x_n$ and AWGN $w_n$ having variance $\sigma_w^2$. The output is thus,
\begin{align*}
y_n&=f(x_n+w_n).
\end{align*}
We assume that we observe both the input sequence $x_n$ as well as the output $y_n$. The function $f$ on the other hand is unknown and we wish to estimate it.

Let us consider first the case where it is known that $f$ is invertible.
A possible means to identify $f$ is as follows. Apply another function $g$ to the output to obtain
\begin{align*}
\hat{z}_n&=g(y_n) \\
&=g(f(z_n)) \\
&=h(z_n),
\end{align*}
where $h=g \circ f$ is the composition of the functions $f$ and $g$.
Assume now that $x_n$ is an AWGN process as well (i.e., a training sequence drawn according to such statistics) with variance $\sigma_x^2$.
We may further assume for simplicity that $\mathbb{E}[h(z_n)]=0$. Define
\begin{align*}
K_1&=\frac{ \mathbb{E}^2 [ h(z_n) x_n ] }  {  \mathbb{E}[h^2(z_n)] \mathbb{E}[x_n^2] }.
\end{align*}
Then, since for bivariate normal random variables nonlinear functions cannot increase (the absolute value of) correlation \cite{lancaster1957some}, it follows
that $K_1$ is maximized (only) when $g=c \cdot f^{-1}$ for some constant $c$, so that $h$ is a linear function. As $K_1$ may be estimated by replacing expectations with time averages, we have obtained a simple criterion for identification of the non-linearity $f$ (up to to a scale factor that may easily be subsequently estimated).

A limitation of the identification criterion described above, is that it does not apply to non-linearities that are not invertible.\footnote{Nonetheless, it's a suitable criterion for memoryless nonlinear compensation, where the inverse of the non-linearity (if exists) is desired (see, e.g.,~\cite{tsimbinos2001nonlinear}).}
We now outline how the inequality derived in this note may serve to overcome this limitation. We note, however, that a drawback of the system described next is that we need to assume
that the signal-to-noise ratio $\sigma_x^2/\sigma_w^2$ at the input of the non-linearity is known, unlike for the system described above.

Let $\alpha=\sqrt{\frac{\sigma_x^2+\sigma_w^2}{\sigma_x^2}}$ so that $\alpha x_n$ has the same variance as $z_n$.
Consider now passing $\alpha x_n$ through a non-linearity $g$ as depicted in Figure~\ref{fig:Fig2},
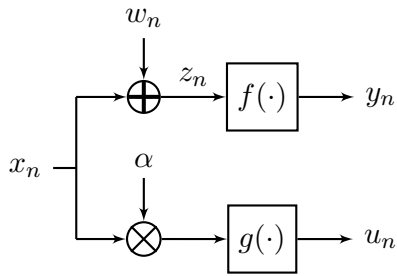
\begin{figure}[h]
\begin{center}
\tikzstyle{int}=[draw, minimum size=2.2em]
\tikzstyle{init} = [pin edge={to-,thin,black}]
\tikzstyle{times} = [draw, circle, inner sep=0.15cm]
\tikzstyle{sum} = [draw, circle, inner sep=0.04cm]
\begin{tikzpicture}[scale=1.05, transform shape, thick, node distance=2cm,auto,>=latex']
    \node (x) at (0,-0.9) {$x_n$};
    \node (xd) [right of=x,node distance=0.645cm, coordinate] {};
    \path[-]  (x) edge (xd);
    \draw[-]  (xd) --++ (0,0.9);
    \draw[-]  (xd) --++ (0,-0.9);

    \node [sum] at (1.5,0) (s) {+};
    \node (s1) at (s) {\scalebox{2}{+}};
    \node (x1) at (0.5,0) {};
    \path[->] (x1) edge (s);
    \node (w) [above of=s,node distance=1cm] {$w_n$};
    \path[->] (w) edge (s);
    \node [int] (f) [right of=s,node distance=1.5cm] {$f(\cdot)$};
    \path[->] (s) edge node {$z_n$} (f);
    \node (y) [right of=f,node distance=1.5cm] {$y_n$};
    \path[->] (f) edge (y);

    \node [times] at (1.5,-1.8) (t) {};
    \node (t1) at (t) {\scalebox{1.8}{$\times$}};
    \node (x2) at (0.5,-1.8) {};
    \path[->] (x2) edge (t);
    \node (sigma) [above of=t,node distance=1cm] {$\alpha$};
    \path[->] (sigma) edge (t);
    \node [int] (g1) [right of=t,node distance=1.5cm] {$g(\cdot)$};
    \path[->] (t) edge (g1);
    \node (u) [right of=g1,node distance=1.5cm] {$u_n$};
    \path[->] (g1) edge (u);
\end{tikzpicture}
\end{center}
\caption{A system for identification of a non-linearity.}
\label{fig:Fig2}
\end{figure} 
to obtain $u_n=g(\alpha x_n)$.
Define
\begin{align*}
K_2&=\frac{ \mathbb{E}^2 [f(z_n)g(\alpha x_n)] } { \mathbb{E}[g(z_n) g(\alpha x_n) ] }.
\end{align*}
It follows from Lemma~\ref{lemma1} that $K_2$ is  maximized \emph{only when} $g=c \cdot f$ for some constant $c$.\footnote{Note that the ``only when" property, which 
is crucial for the identification problem, follows by the condition for equality in Lemma~\ref{lemma1}.}  Again, $K_2$ may be computed by replacing expectations with time averages.
This is the case, as although we do not observe $w_n$, we can replace it with AWGN noise $w'_n$ generated with the same variance, to compute
\begin{align*}
\mathbb{E}[g(z_n) g(\alpha x_n)]=\mathbb{E}[g(x_n+w'_n) g(\alpha x_n)].
\end{align*}
We note that for a practical implementation of the scheme, one would need some explicit (parametric) representation for the function $g(\cdot)$. For instance, one could employ 
a series expansion in orthogonal polynomials (see, e.g., \cite{tsimbinos2001nonlinear}).
\section*{Acknowledgment}
The authors are grateful to Or Ordentlich for helpful discussions concerning the contents of Section~\ref{sec:app}.

\bibliographystyle{elsarticle-num}
%




%

\end{document}